%% file: main.tex
\documentclass[letterpaper,10pt,conference]{ieeeconf}

\usepackage{cite}
\usepackage{graphicx}
\usepackage{subcaption}
\usepackage[cmex10]{amsmath}
\usepackage{algpseudocode}
\usepackage{algorithmicx} 
\usepackage{array}
\usepackage{url}

\usepackage{todonotes}

\usepackage{epsfig}
\usepackage{amsfonts,amssymb,multirow,bigstrut,booktabs,ctable,latexsym}

\usepackage[mathscr]{eucal}
\usepackage{verbatim}

\usepackage{amsthm}
\usepackage{algorithm}

\IEEEoverridecommandlockouts  
\begin{document}

\newtheorem{definition}{Definition}
\newtheorem{lemma}{Lemma}
\newtheorem{theorem}{Theorem}
\newtheorem{example}{Example}
\newtheorem{proposition}{Proposition}
\newtheorem{remark}{Remark}
\newtheorem{assumption}{Assumption}
\newtheorem{corrolary}{Corrolary}
\newtheorem{property}{Property}
\newtheorem{ex}{EX}
\newtheorem{problem}{Problem}
\newcommand{\argmin}{\arg\!\min}
\newcommand{\argmax}{\arg\!\max}
\newcommand{\st}{\text{s.t.}}
\newcommand \dd[1]  { \,\textrm d{#1}  }

\title{\Large\bf Minimum Violation Control Synthesis on Cyber-Physical Systems under Attacks}

\author{Luyao Niu, Jie Fu and Andrew Clark %
\thanks{L. Niu, J. Fu and A. Clark are with the Department of Electrical and Computer Engineering, Worcester Polytechnic Institute, Worcester, MA 01609 USA.
{\tt\small \{lniu,jfu2,aclark\}@wpi.edu}}
	\thanks{This work was supported by NSF grant CNS-1656981.}
}
\thispagestyle{empty}
\pagestyle{empty}

\maketitle

\begin{abstract}

Cyber-physical systems are conducting increasingly complex tasks, which are often modeled using formal languages such as temporal logic. The system's ability to perform the required tasks can be curtailed by malicious adversaries that mount intelligent attacks. At present, however, synthesis in the presence of such attacks has received limited research attention. In particular, the problem of synthesizing a controller when the required specifications cannot be satisfied completely due to adversarial attacks has not been studied. In this paper, we focus on the minimum violation control synthesis problem under linear temporal logic constraints of a stochastic finite state discrete-time system with the presence of an adversary. A minimum violation control strategy is one that satisfies the most important tasks defined by the user while violating the less important ones. We model the interaction between the controller and adversary using a concurrent Stackelberg game and present a nonlinear programming problem to formulate and solve for the optimal control policy. To reduce the computation effort, we develop a heuristic algorithm that solves the problem efficiently and demonstrate our proposed approach using a numerical case study.
\end{abstract}

\input{introduction}

\input{related}

\input{preliminary}
\input{formulation}

\input{unreal_solution}

\input{simulation}
\input{conclusion}

\bibliographystyle{IEEEtran}
\bibliography{IEEEabrv,MyBib}

\end{document}

%% file: introduction.tex
\section{Introduction}\label{sec: intro}

Cyber-physical systems have been identified to play important roles in multiple application domains such as health care systems, cloud computing, and smart homes. To model the increasingly complex tasks and corresponding desired system behaviors consistently, rigorously and compactly, temporal logics such as linear temporal logic (LTL) and computation tree logic (CTL) are adopted in recent literature. Typical system properties that can be modeled using LTL, whose syntax and semantics have been well developed, include liveness (e.g., `always eventually A'), reactivity (e.g., `if A, then B'), safety (e.g., `always not A') and so on.

Formal methods provide a class of theory and methods for controller design to satisfy given specifications modeled using temporal logics. Such control synthesis problems have been investigated in different applications such as robotic motion planning \cite{bhatia2010sampling,lahijanian2012temporal} and optimal control \cite{wolff2012robust,ding2014optimal}. However, these works normally explicitly or implicitly assume the existence of the controller, which is not always the case. 

In \cite{raman2012automated}, unsynthesizable controllers are characterized as either unsatisfiability or unrealizability. Unsatisfiability is caused by the incompatibility of the specifications given to the system, while unrealizability is caused by uncertainties and stochastic errors. Different from uncertainties and stochastic errors, malicious attacks can also cause unsynthesizable controllers in CPS. Malicious attacks on CPS raise the concern on CPS security since they can lead to misbehaviors and failures. For instance, power outage caused by attackers on power system \cite{CIAreport}, a false data injection (FDI) based attack CarShark on automobiles \cite{koscher2010experimental} and widely known Stuxnet on industrial control system (ICS) all caused significant economic losses and/or safety risks. 

The approaches proposed for analyzing uncertainties and stochastic errors are not applicable for analyzing malicious attacks on CPS. Moreover, uncertainties and stochastic errors are often viewed as identically and independently distributed random variables, which is not the case for malicious and strategic attacks. In the worst case, stochastic elements such as environment behavior are interpreted as malicious attacks on the system. Zero-sum game provides a good model for worst case analysis \cite{fudenberg1991game}. Meanwhile, failures returned by control synthesis framework could also be caused by malicious and strategic attacks such as jamming attack and Denial-of-Service (DoS) attack which are subject to different information pattern comparing to zero-sum game. In security domain, Stackelberg game is a more reasonable model \cite{tambe2011security,zhu2011stackelberg}, where player $1$ (always denoted as leader in the game) commits to its strategy first and player $2$ (always denoted as follower in the game) observes leader's strategy and then plays its best response. Stackelberg game can capture the information asymmetry and model the value of information.

In this paper, we consider a stochastic discrete-time system with the presence of an adversary, which is abstracted as a stochastic game (SG). The system is given a set of specifications modeled in LTL co-safe (scLTL). We focus on the scenario where no controller can be synthesized to satisfy the specifications simultaneously due to either incompatibility between specifications or the presence of the adversary. Thus we aim at the minimum violation control strategy synthesis problem, i.e., compute a control strategy that violates the less important specifications and satisfies the most important specifications based on user's preference \cite{tuumova2013minimum}. To the best of our knowledge, this is the first attempt to analyze minimum violation control synthesis on stochastic system in the presence of adversary. To summarize, we make the following contributions. We formulate a stochastic game to model the interaction between the controller and adversary. We give examples for typical attacks in CPS that can be incorporated into our proposed framework. To model limited observation capability of human adversary, anchoring bias is considered. We present the completion procedure to augment each automaton associated with each specification given to the system. We calculate the product SG using the completed automaton and SG. We formulate a nonlinear programming problem on the product SG to calculate the optimal control policy. A heuristic algorithm is proposed to compute an approximate solution. The proposed algorithm significantly saves computation cost and memory cost. The convergence of the algorithm is proved. A numerical case study is used to demonstrate the proposed approach. By using the proposed approach, more specifications can be satisfied when considering the presence of the adversary. Finally, we show the relationship between the controller's expected utility and the anchoring bias parameter of adversary. 

The remainder of this paper is organized as follows. Related work is presented in Section \ref{sec: related} and preliminary backgrounds are presented in Section \ref{sec: preliminary}. Section \ref{sec: formulation} presents problem formulation. We give solution method in Section \ref{sec: solution}. A numerical case study is given in Section \ref{sec: simulation}. We conclude this paper in Section \ref{sec: conclusion}.

%% file: related.tex
\section{Related Work}\label{sec: related}

Control synthesis under temporal logic constraints normally assumes the specifications can be satisfied. Contributions on the cases when the specifications cannot be fulfilled can be classified into four categories. First, the minimum violation problem for deterministic system has been studied in \cite{tuumova2013minimum,tumova2016least, vasile2017minimum,chaudhari2014incremental}. Violations caused by confliction between specifications have been studied in \cite{tuumova2013minimum,tumova2016least, vasile2017minimum}, and a control strategy that satisfies the most important specifications is synthesized. In \cite{chaudhari2014incremental}, a two-player concurrent Stackelberg differential game is formulated. Quantitative preference over satisfactions of scLTL is investigated in \cite{lahijanian2015time}. However, contributions \cite{tuumova2013minimum,tumova2016least, vasile2017minimum,chaudhari2014incremental,lahijanian2015time} focus on deterministic systems and hence the proposed approaches are not applicable to stochastic systems. Second, unsynthesizable specifications are analyzed in \cite{raman2012automated}. Third, model repair problem is investigated so that satisfaction on specifications is guaranteed \cite{buccafurri1999enhancing,bartocci2011model}. Finally, specification revision problem is investigated in \cite{kim2015minimal}. Planning revision under temporal logic specification is investigated in \cite{guo2013reconfiguration}. However, none of the aforementioned papers consider the presence of adversary. Furthermore, non-deterministic automata are used in the aforementioned papers while deterministic automata are used in this paper.

Secure control in adversarial environment has been investigated using both control theoretic based approach \cite{shoukry2017secure} and game theoretic methods \cite{zhu2011stackelberg,zhu2015game}. When game theory meets temporal logic, turn-based two-player SG has been used to construct model checker \cite{chen2013prism} and model checking framework \cite{quatmann2016parameter, kattenbelt2010game}. The difference is that a general sum concurrent SG is considered in this paper. Secure control under LTL formula specification modeling liveness and safety constraints is considered in \cite{niu2018Secure}. The proposed approach in \cite{niu2018Secure} focuses on liveness and safety constraints, while this paper considers specifications modeled using scLTL.

%% file: preliminary.tex
\section{Preliminaries}\label{sec: preliminary}

In this section, we present backgrounds on linear temporal logic and stochastic games.

\subsection{Linear Temporal Logic (LTL)}

An LTL formula consists of a set of atomic propositions $\Pi$, boolean operators including negation ($\neg$), conjunction ($\land$) and disjunction ($\lor$) and temporal operators including next ($X$) and until ($\mathcal{U}$) \cite{baier2008principles}. An LTL formula is defined inductively as
\begin{equation*}
\phi=True\mid\pi\mid\neg\phi\mid\phi_1\land\phi_2\mid X\phi\mid\phi_1~\mathcal{U}~\phi_2.
\end{equation*}
Other operators such as implication ($\implies$), eventually ($F$) and always ($G$) can be defined using operators above. In particular, $\phi\implies\psi$ is equivalent to $\neg\phi\lor\psi$, $F \phi$ is equivalent to $True~\mathcal{U}~\phi$, and $G\phi$ is equivalent to $\neg F\neg\phi$.

The semantics of LTL formulas are defined over infinite words in $2^{\Pi}$. Informally speaking, $G\phi$ is true if and only if $\phi$ is true for the current time step and all future time. $F\phi$ is true if and only if $\phi$ is true at some future time. $X\phi$ is true if and only if $\phi$ is true in the next time step. A word $\eta$ satisfying an LTL formula $\phi$ is denoted as $\eta\models\phi$.

In this paper, we focus on syntactically co-safe LTL (scLTL) formulas.

\begin{definition}\label{def: scLTL}
\emph{(scLTL \cite{kupferman2001model}):} Any string that satisfies a scLTL formula consists of a finite string (a good prefix) followed by any infinite continuation. This continuation does not affect the formula's truth value.
\end{definition}
By Definition \ref{def: scLTL}, a word $\eta$ satisfies an scLTL formula $\phi$ if it contains a good prefix $\eta_0\eta_1\cdots\eta_n$ such that $\eta_0\eta_1\cdots\eta_n\eta_{n+1}\eta_{n+2}\cdots\models\phi$ for any suffix $\eta_{n+1}\eta_{n+2}\cdots$.  

For each scLTL formula, a deterministic finite automaton (DFA) can be obtained. A DFA is defined as follows.
\begin{definition}\label{def: DFA}
\emph{(Deterministic finite automaton):} A DFA $\mathcal{A}$ is a tuple $\mathcal{A}=(Q,q_0,\Sigma,\delta,F)$, where $Q$ is a finite set of states, $q_0\in Q$ is the initial state, $\Sigma$ is alphabet, $\delta: Q\times\Sigma\rightarrow Q$ is the set of transitions and $F\subseteq Q$ is the set of accepting states. 
\end{definition}

A run on a DFA $\mathcal{A}$ over a finite input word $\sigma = \sigma_0\sigma_1\cdots\sigma_n$ is a sequence of states $Q^*=q_0q_1\cdots q_n$ such that $\delta(q_{k-1},\sigma_k)=q_k$ for all $0\leq k\leq n$. A run is accepting if $q_n\in F$. The satisfaction of a formula $\phi$ by a run $\sigma$ is denoted as $\sigma\models\phi$. To enable violations on specifications, we assume any DFA $\mathcal{A}$ is complete, i.e., for any $q\in Q$ and $\sigma\in\Sigma$, $\delta(q,\sigma)$ is defined. The completion procedure can be achieved by adding an additional $sink$ state and let $\delta(q,\sigma)=sink$ if $\delta(q,\sigma)$ is undefined.




\subsection{Stochastic Game (SG)}
A Stochastic Game (SG) is defined as follows.
\begin{definition}\label{def: SG}
\emph{(Stochastic Game):} A stochastic game $\mathcal{G}$ is a tuple $\mathcal{G}= (S,U_C,U_A,Pr,\Pi,\mathcal{L})$, where $S$ is a finite set of states, $U_C$ is a set of actions of the controller, $U_A$ is a set of actions of an adversary, $Pr: S\times U_C\times U_A\times S\rightarrow [0, 1]$ is a transition function where $Pr(s, u_C, u_A, s^{\prime})$ is the probability of a transition from state $s$ to state $s^\prime$ when the controller takes action $u_C$ and the adversary takes action $u_A$. $\Pi$ is a set of atomic propositions. $\mathcal{L}:S\rightarrow 2^{\Pi}$ is a labeling function mapping each state to a subset of propositions in $\Pi$.
\end{definition}
Denote the admissible actions as the set of actions available to the controller (resp. adversary) at each state $s\in S$ as $U_C(s)$ (resp. $U_A(s)$). A finite (resp. infinite) path on SG $\mathcal{G}$ is a finite (resp. infinite) sequence of states denoted as $Path_{fin} = s_0s_1\cdots s_n$(resp. $Path_{inf} = s_0s_1s_2\cdots$). Let $Path$ be the set of finite paths. A control policy $\mu:Path\times U_C\rightarrow \mathbb{R}$ (resp. adversary policy $\lambda:Path\times U_A\rightarrow \mathbb{R}$) is a function specifying the probability distribution over control (resp. attack) actions given historical trajectory $Path_{fin}$. An admissible policy is the policy whose support is the set of admissible actions at each state. In particular, we consider a \emph{memoryless} control(resp. adversary) policy in this paper, i.e., $\mu$ (resp. $\lambda$) depends only on the current state. 

Stackelberg SG is a widely adopted model in security domain \cite{tambe2011security}. In the Stackelberg setting, one player is the leader and another player is the follower. The leader first commits to a strategy $\mu$. The follower then observes the strategy $\mu$ and play its best response $\lambda$. Given any control policy $\mu$, the best response from the adversary is represented as $\mathcal{BR}(\mu)=\{\lambda|\lambda=\argmax_{\lambda}\mathcal{T}_A(\mu,\lambda)\}$, where $\mathcal{T}_A(\mu,\lambda)$ is the adversary's utility given a pair of leader-follower strategies. The Stackelberg equilibrium is defined as follows.
\begin{definition}\label{def: Stackelberg equilibrium}
\emph{(Stackelberg Equilibrium (SE)):} Denote the utility that the leader (resp. follower) gains in a stochastic game $\mathcal{G}$ under leader follower strategy pair $(\mu,\lambda)$ as $\mathcal{T}_C(\mu,\tau)$ (resp. $\mathcal{T}_A(\mu,\tau)$). A pair of leader follower strategy $(\mu,\lambda)$ is an SE if leader's strategy $\mu$ is optimal given that the follower observes its strategy and plays its best response, i.e., $\mu\in\argmax_{\mu^\prime\in\boldsymbol{\mu}}~ \mathcal{T}_C(\mu^\prime,\mathcal{BR}(\mu^\prime))$,
where $\boldsymbol{\mu}$ is the set of all admissible policies of the controller and $\lambda\in\mathcal{BR}(\mu^\prime)$ denotes the best response to the leader's strategy $\mu^\prime$ from the follower.
\end{definition}

In Stackelberg games with human adversaries, \emph{anchoring bias} is used to model the confidence of the adversary in its observations on $\mu$ \cite{fox2003partition}. When considering anchoring bias, the response $\lambda$ might not be the best response to control policy $\mu$. Human adversaries normally assign uniform probability to the control action at each state \cite{fox2003partition}. When more information is obtained via observation, adversaries slowly update the distributions. In this paper a linear model is adopted to represent the estimated probability. In this model, the estimated probability of human adversary that the controller takes action $u_C$ at each state is calculated as 
\begin{equation}\label{eq: anchoring bias}
\tilde{\mu}(s,u_C)=\alpha\frac{1}{|U_C(s)|} + (1-\alpha)\mu(s,u_C),~\forall s, u_C
\end{equation}
where $\alpha\in[0,1]$ is a parameter to tune the balance between the original and true probability. When $\alpha = 0$, the estimated probability becomes the true probability and thus the adversary plays its best response. When $\alpha=1$, then the estimated probability becomes the uniform distribution, implying the adversary has no capability to observe or infer the control policy based on his observation.

%% file: formulation.tex
\section{Problem Formulation}\label{sec: formulation}

In this section, we first present the problem formulation. Then we show that several typical CPS security problems can be analyzed using the proposed framework. We consider a finite-state discrete-time system in the presence of an adversary, which can be abstracted using a SG $\mathcal{G}_0=(S,U_C,U_A,Pr,\Pi,\mathcal{L})$ as defined in Definition \ref{def: SG}.


We adopt the concurrent Stackelberg setting. In particular, the controller acts as the leader and the adversary is the follower. The controller first commits to its strategy (or control policy) $\mu$. Then the adversary, who observes the historical behavior of the controller, plays its response $\lambda$ to the control policy $\mu$. We assume that both the controller and adversary can observe current state $s$. At each system state $s$, both the controller and adversary have to take actions simultaneously and the system evolves to state $s^\prime$ following transition function defined in Definition \ref{def: SG}. 

The system is assigned a set of specifications $\Phi=\{\phi_1,\phi_2,\cdots,\phi_n\}$ modeled using scLTL \cite{tumova2016least,vasile2017minimum}. By satisfying each specification $\phi_i\in\Phi$, the controller gains a reward $r(\phi_i)$. The objective of the controller is to maximize the total reward obtained via satisfying specifications. In the worst case, the adversary attempts to deviate system behavior and drive the system to violate specifications in $\Phi$ so as to minimize the total reward obtained by system. Hence, the specifications cannot be satisfied simultaneously due to either incompatibility of specifications or the presence of adversary. Thus we investigate the minimum violation problem on such a system as follows.
\begin{problem}\label{prob: unrealizability problem}
Given an SG $\mathcal{G}_0$ abstracted from the system in the presence of an adversary and a set of specifications $\Phi=\{\phi_1,\cdots,\phi_n\}$ that potentially cannot be satisfied by system simultaneously, with each $\phi_i\in\Phi$ associated with a reward function $r(\phi_i)$, compute a control policy $\mu$ such that $\mu$ and the best response from adversary $\lambda\in\mathcal{BR}(\tilde{\mu})$ constitutes SE defined in Definition \ref{def: Stackelberg equilibrium}.
\end{problem}

In the following, we show several problems in security domain can be formulated using our proposed framework.

\subsubsection{Patrolling Security Game with single type of adversary \cite{basilico2012patrolling}}\label{ex: security game}


The states $S$ are set as locations in PSG. The actions $U_C$ and $U_A$ are the actions available to the patrol unit and adversary, respectively. In particular, $U_C$ includes the actions that transit the patrol unit among the locations, while $U_A$ are the intrusion actions modeling which location is targeted by the adversary. The transition probability captures the transition uncertainty. The actions taken by both players jointly determine their utilities. For instance, the adversary wins if the target region is under attack without protection and the patrol unit wins otherwise. 

The interaction between the patrol unit and adversary is modeled as a Stackelberg game. The security force is the leader while the adversary is the follower. The adversary can observe the schedule of security force (by waiting outside the environment indefinitely) and play its best response.

By using our proposed framework, task dependent rewards can be defined and thus more complex behaviors of the patrolling unit can be considered. For example, the patrolling unit can be given the following tasks: visit areas in sequence (e.g., `First region A then region B then region C': $F~(A\land (F~B\land F~C))$) and reactivity (e.g., `if some passenger enters prohibited region, stop them': $prohibited\implies stop$).

\subsubsection{Jamming Attacks on CPS}\label{ex: jamming attack}
Applications such as SCADA networks and remotely controlled UAVs can be modeled as CPS where the controller communicates with the plant via a wireless network corrupted by a strategic jamming attacker.

Let the state of the plant evolves following a finite state discrete-time dynamics $x(k+1)=Ax(k)+Bu(k)+\omega(k),~k=0,1\cdots$, where $x(k)$ is the system state, $u(k)=\Gamma(u_C(k),u_A(k))$ is the system input jointly determined by the control signal $u_C(k)$ and adversary signal $u_A(k)$ for all $k$ and $\omega(k)$ is stochastic disturbance. Function $\Gamma(u_C(k),u_A(k))$ can be formulated as: (i)
$\Gamma(u_C(k),u_A(k))=u_A(k)\cdot u_C(k),~ u_A(k)\in\{0,1\},~\forall k$ \cite{gupta2010optimal}, or (ii) $\Gamma(u_C(k),u_A(k))=u_A(k)+u_C(k),~\forall k$ \cite{li2007optimal}. The formulation of (i) models scenario where the adversary can cause collision at the receiver equipped on the plant and result in denial-of service (DoS) attack. The formulation in (ii) models the scenario where the adversary can flip several bits in the packet and result in false information at the plant. Note that when the adversary launches DoS attack, the actuator can generate no input for the plant as $u(k)=0$ when $u_A(k)=0$ \cite{gupta2010optimal}, or $u(k)=u(k-1)$ when $u_A(k)=0$.

Consider the example of an autonomous UAV. The reachability specification can be given to the UAV as `eventually reach target region and avoid obstacles', i.e, $G\neg~obstacle~\land~ F~target$. 


%% file: unreal_solution.tex
\section{Proposed Solution for Problem \ref{prob: unrealizability problem}}\label{sec: solution}

In this section, we first present a mixed integer non-linear programming (MINLP) formulation. Then we propose a heuristic solution to compute a \emph{proper} stationary control policy, which will be defined later. 

For each specification $\phi_i$, a complete DFA $\mathcal{A}_i=(Q^i,q_0^i,\Sigma,\delta^i,F^i)$ can be constructed. Given the set of complete automata $\mathcal{A} = \{\mathcal{A}_1,\cdots,\mathcal{A}_n\}$ with each $\mathcal{A}_i$ associated with $\phi_i$, we can construct a product automaton using the following definition \cite{baier2008principles}.
\begin{definition}\label{def: product automaton}
\emph{(Product automaton):} A product automaton obtained from $\mathcal{A}$ is a tuple $\mathcal{A}_\mathcal{P}=(Q_\mathcal{P},q_{0,\mathcal{P}},\Sigma,\delta_\mathcal{P},F_\mathcal{P})$, where $Q_\mathcal{P}=Q^1\times\cdots\times Q^n$ is a finite set of states, $q_{0,\mathcal{P}}=(q_0^1,\cdots,q_0^n)$ is the initial state, $\Sigma$ is the alphabet inherited from $\mathcal{A}$, $\delta_\mathcal{P}=\left((q^1,\cdots,q^n),\sigma,({q^{1}}^\prime,\cdots,{q^{n}}^\prime)\right)$ if $\delta^i(q^i,\sigma)={q^i}^\prime$ for all $i\}$ and $F_\mathcal{P}=\left\{(q^1,\cdots,q^n)|q^i\in F^i,~\forall i\right\}$ is the set of accepting states.
\end{definition}

Given the SG $\mathcal{G}_0$ and product automaton $\mathcal{A}_\mathcal{P}$, we can construct a product SG $\mathcal{G}$ defined as follows.
\begin{definition}\label{def: product SG}
\emph{(Product SG):} Given SG $\mathcal{G}_0 = (S,U_C,U_A,Pr,\mathcal{L},\Pi)$ and product automaton $\mathcal{A}_\mathcal{P}=(Q_\mathcal{P},q_{0,\mathcal{P}},\Sigma,\delta_\mathcal{P},F_\mathcal{P})$, a (weighted and labeled) product SG is a tuple $\mathcal{G}=(S_\mathcal{P},U_C,U_A,Pr_\mathcal{P},Acc,W)$, where $S_\mathcal{P}=S\times Q_\mathcal{P}$ is a finite set of states, $U_C$ (resp. $U_A$) is a finite set of control inputs (resp. attack signals), $Pr_\mathcal{P}((s,q^1,\cdots,q^n),u_C,u_A,(s^\prime,{q^1}^\prime,\cdots,{q^n}^\prime))=Pr(s,u_C,u_A,s^\prime)$ if $((q^1,\cdots,q^n),\mathcal{L}(s^\prime),({q^1}^\prime,\cdots,{q^n}^\prime))\in\delta_\mathcal{P}$, $Acc = S\times F_\mathcal{P}$, and $W$ is a weight function assigning each transition a reward.
\end{definition}
The weight function of product SG $\mathcal{G}$ is defined as
\begin{equation}\label{eq: weight function}
W((s,q^1,\cdots,q^n),u_C,u_A,(s^\prime,{q^1}^\prime,\cdots,{q^n}^\prime))\\
=\sum_{i=1}^n I_{ii^\prime}r(\phi_i),
\end{equation}
where the indicator $I_{ii^\prime} = 1$ if $q^i\notin F^i$ and ${q^i}^\prime\in F^i$ and $I_{ii^\prime} = 0$ otherwise. By the definition \eqref{eq: weight function}, we have that a trace $\tau$ collects reward by satisfying specifications if a specification is satisfied at first time. We index the states in the product SG $\mathcal{G}$ as $s_\mathcal{P}$. 

A proper control policy on product SG is defined as follows.
\begin{definition}\label{def: proper policy}
\emph{(Proper Policies):} A stationary control policy $\mu$ is proper if under $\mu$, regardless of the policy chosen by the adversary, the set of destination states can eventually be reached with positive probability, where a destination state $s_\mathcal{P}=(s,q^1,\cdots,q^n)$ is a state such that $q^i$ is an absorbing state in automaton $\mathcal{A}_i$ for all $i$.  
\end{definition}
If a control policy $\mu^\prime$ is improper, then under policy $\mu^\prime$, there exists some state $s_\mathcal{P}$ that has zero probability to reach the set of destination states.
\subsection{MINLP Formulation}

For the controller's strategy, since randomized stationary strategies are considered in Problem \ref{prob: unrealizability problem}, we have that 
\begin{align}
&\mu(s_\mathcal{P},u_C)\geq 0,~\forall s_\mathcal{P}\in S_\mathcal{P}, u_C\in U_C,\label{eq: constraint 1}\\
&\sum_{u_C\in U_C(s_\mathcal{P})}\mu(s_\mathcal{P},u_C)=1,~\forall s_\mathcal{P}\in S_\mathcal{P},\label{eq: constraint 2}\\
&\lambda(s_\mathcal{P},u_A)\in\{0,1\},~\forall s_\mathcal{P}\in S_\mathcal{P}, u_A\in U_A,\label{eq: constraint 3}\\
&\sum_{u_A\in U_A(s_\mathcal{P})}\lambda(s_\mathcal{P},u_A)=1,~\forall s_\mathcal{P}\in S_\mathcal{P},\label{eq: constraint 4}
\end{align}
where \eqref{eq: constraint 2} and \eqref{eq: constraint 4} guarantees that the probability distribution sums to one. Eq. \eqref{eq: constraint 3} holds since in Stackelberg games, it is sufficient to consider pure strategies for the follower \cite{vorobeychik2012computing}.



The value function for the controller $V_C(s_\mathcal{P})$ (resp. adversary $V_A(s_\mathcal{P})$) is defined as the expected reward for the controller (resp. adversary) starting from state $s_\mathcal{P}$. The value functions can be characterized using the following lemma.
\begin{lemma}\label{lemma: value bound}
The expected reward of the controller and adversary induced by policy $\mu$ and $\lambda$ can be represented as
\begin{align*}
&V_C(s_\mathcal{P})=\sum_{u_C\in U_C}\big[\mu(s_\mathcal{P},u_C)\sum_{u_A\in U_A}\lambda(s_\mathcal{P},u_A)\nonumber\\
&~\sum_{s_\mathcal{P}^\prime}Pr_\mathcal{P}(s_\mathcal{P},u_C,u_A,s_\mathcal{P}^\prime)(W(s_\mathcal{P},u_C,u_A,s_\mathcal{P}^\prime)+V_C(s_\mathcal{P}^\prime))\big],\\    
&V_A(s_\mathcal{P})=\sum_{u_C\in U_C}\big[\tilde{\mu}(s_\mathcal{P},u_C)\sum_{u_A\in U_A}\lambda(s_\mathcal{P},u_A)\nonumber\\
&~\sum_{s_\mathcal{P}^\prime}Pr_\mathcal{P}(s_\mathcal{P},u_C,u_A,s_\mathcal{P}^\prime)(-W(s_\mathcal{P},u_C,u_A,s_\mathcal{P}^\prime)+V_A(s_\mathcal{P}^\prime))\big].
\end{align*}
Moreover, given a pair of policies $\mu$ and $\lambda$, the expected reward of the controller and adversary are the unique solutions to the linear equations above.
\end{lemma}
\begin{proof}
The expected reward starting from state $s_\mathcal{P}$ is calculated as
\begin{multline*}
\tilde{W}(s_\mathcal{P})= \sum_{u_C\in U_C}\mu(s_\mathcal{P},u_C)\sum_{u_A\in U_A}\lambda(s_\mathcal{P},u_A)\\
\sum_{s_\mathcal{P}^\prime}Pr_\mathcal{P}(s_\mathcal{P},u_C,u_A,s_\mathcal{P}^\prime)W(s_\mathcal{P},u_C,u_A,s_\mathcal{P}^\prime).   
\end{multline*}

Given a pair of policies $\mu$ and $\lambda$, the stochastic game reduces to a Markov chain. The expected reward $V_C(s_\mathcal{P})$ is then viewed as the expected reward collected by the path starting from $s_\mathcal{P}$ to the set of destination states, which is equivalent to the shortest path problem on the induced Markov chain. By the dynamic programming algorithm of stochastic shortest path problem on Markov chain, we have\cite{bertsekas1995dynamic}
\begin{multline}\label{eq: VC Bellman}
V_C(s_\mathcal{P})= \tilde{W}(s_\mathcal{P})\\
+\sum_{s_\mathcal{P}^\prime}Pr_\mathcal{P}(s_\mathcal{P},u_C,u_A,s_\mathcal{P}^\prime)V_C(s_\mathcal{P}^\prime).  
\end{multline}
Then by the definition of $\tilde{W}$, we have that Lemma \ref{lemma: value bound} holds.
\end{proof}
Based on Lemma \ref{lemma: value bound}, we have the following proposition which gives the sufficiency of considering proper policies.

\begin{proposition}\label{prop: sufficiency of proper policy}
If a proper control policy $\mu^\prime$ is associated with the highest expected reward for the controller among all proper policies, then it associates with the highest expected reward among all stationary policies.
\end{proposition}
\begin{proof}
Let $\mu$ be the control policy that enables the controller receiving highest reward among all stationary policies. If $\mu$ is a proper policy, then the result clearly holds. Next, we focus on the scenario where $\mu$ is improper and show that by construction, we have a proper control policy $\mu^\prime$ such that the expected rewards for the controller under policy $\mu$ and $\mu^\prime$ are equal. Divide the set of states $S_\mathcal{P}$ into two subsets $S_1$ and $S_2$. In particular, let $S_1$ be the set of states that cannot reach the set of destination states, while $S_2$ denotes the set of states that reach the set of destination states with positive probability. Let $\mu^\prime=\mu$ for all $s_\mathcal{P}\in S_2$. By hypothesis on $\mu$, we have that the proper control policy $\mu^\prime$ corresponds to the highest expected reward when the initial state is in $S_2$. By the assumption on the existence of a proper policy $\tilde{\mu}$, we let $\mu^\prime=\tilde{\mu}$ for all $s_\mathcal{P}\in S_1$. Since $\mu$ is an improper policy while $\tilde{\mu}$ is a proper policy, we have that the expected reward received by the controller by committing to control policy $\mu^\prime$ is no less than committing to $\mu$. Hence, we have a proper control policy $\mu^\prime$ such that the controller receives expected reward no less than committing to improper policy $\mu$.  
\end{proof}

By Proposition \ref{prop: sufficiency of proper policy}, we can restrict the search space of control policy to the set of proper control policies. Denote the expected reward obtained by the controller starting from state $s_\mathcal{P}$ when the controller commits to strategy $\mu$ and adversary takes action $u_A$ as $B_C(s_\mathcal{P},\mu,u_A)$. Define $B_A(s_\mathcal{P},\tilde{\mu},u_A)$ for the adversary analogously. Then for all $s_\mathcal{P}\in S_\mathcal{P},u_A\in U_A$, the expected reward for the controller (resp. adversary) can be represented as
\begin{align}
&B_C(s_\mathcal{P},\mu,u_A)=\sum_{u_C\in U_C}\mu(s_\mathcal{P},u_C)\big[\sum_{s_\mathcal{P}^\prime}Pr_\mathcal{P}(s_\mathcal{P},u_C,u_A,s_\mathcal{P}^\prime)\nonumber\\
&\quad\quad\quad\quad(W(s_\mathcal{P},u_C,u_A,s_\mathcal{P}^\prime)+V_C(s_\mathcal{P}^\prime))\big],\label{eq: BC}\\
&B_A(s_\mathcal{P},\tilde{\mu},u_A)=\sum_{u_C\in U_C}\tilde{\mu}(s_\mathcal{P},u_C)\big[\sum_{s_\mathcal{P}^\prime}Pr_\mathcal{P}(s_\mathcal{P},u_C,u_A,s_\mathcal{P}^\prime)\nonumber\\
&\quad\quad\quad\quad(-W(s_\mathcal{P},u_C,u_A,s_\mathcal{P}^\prime)+V_A(s_\mathcal{P}^\prime))\big],\label{eq: BA} 
\end{align}
which are the expected utility of the controller and adversary, respectively. Note that the adversary's expected reward depends on its observation over the control policy $\tilde{\mu}$ defined in \eqref{eq: anchoring bias}. Since $\lambda$ is binary, we can bound the values for the adversary and controller using the big M method \cite{vorobeychik2012computing}, respectively, for all $s_\mathcal{P}$ and $u_A$ as follows: 
\begin{align}
&B_A(s_\mathcal{P},\tilde{\mu},u_A)\leq V_A(s_\mathcal{P})\leq B_A(s_\mathcal{P},\tilde{\mu},u_A)\nonumber\\
&\quad\quad\quad\quad\quad\quad\quad\quad\quad\quad\quad\quad+(1-\lambda(s_\mathcal{P},u_A))Z,\label{eq: constraint 5}\\
&V_C(s_\mathcal{P})\leq B_C(s_\mathcal{P},\mu,u_A)+(1-\lambda(s_\mathcal{P},u_A))Z,\label{eq: constraint 7}
\end{align}
where $Z$ is a sufficiently large positive number. Inequality \eqref{eq: constraint 5} and \eqref{eq: constraint 7} give bounds for $V_A(s_\mathcal{P})$ and $V_C(s_\mathcal{P})$. Depending on the value of $\lambda$, the upper bounds for $V_C(s_\mathcal{P})$ (resp. $V_A(s_\mathcal{P})$) can be either infinity ($\lambda(s_\mathcal{P},u_A)=1$) or $B_C(s_\mathcal{P},\mu,u_A)$ (resp. $B_A(s_\mathcal{P},\mu,u_A)$).

To compute the control policy that maximizes the expected utility of controller, the following optimization problem can be formulated \cite{vorobeychik2012computing}.
\begin{align}
\underset{\mu,\lambda,V_C,V_A}{\max}\quad&\gamma^TV_C\label{eq: MINLP}\\
\st\quad&\eqref{eq: anchoring bias}~\eqref{eq: constraint 1}~\eqref{eq: constraint 2}~\eqref{eq: constraint 3}~\eqref{eq: constraint 4}~\eqref{eq: constraint 5}~\text{and}~\eqref{eq: constraint 7}\nonumber
\end{align}
where
$\gamma$ is the initial distribution over state space $S_\mathcal{P}$. Since constraints \eqref{eq: constraint 5} and \eqref{eq: constraint 7} introduce nonlinearity and $\lambda$ is binary, the optimization problem \eqref{eq: MINLP} is an MINLP.

\subsection{Heuristic Solution}

The MINLP \eqref{eq: MINLP} is nonconvex and solving it is NP-hard. In the following, we present a value iteration based heuristic solution to the MINLP \eqref{eq: MINLP}.

\begin{center}
  	\begin{algorithm}[!htp]
  		\caption{Algorithm for computing a control strategy $\mu$ that maximizes the expected reward $V_C$.}
  		\label{alg: HVI}
  		\begin{algorithmic}[1]
  			\State Let $\mathcal{H}\leftarrow \{V_{C,1},\cdots,V_{C,t},\cdots,V_{C,|\mathcal{H}|}\}$, $\mathcal{V}\leftarrow\emptyset$,
  			\For {$V_{C,t}\in\mathcal{H}$}
  			\State $k\leftarrow 0$
  			\Repeat \State Solve MILP \eqref{eq: MILP} to obtain expected reward $V_C^{k}$.
  			\State $k \leftarrow k+1$
  			\Until $\gamma^TV_C^k-\gamma^TV_C^{k-1}\leq\epsilon$ or MILP \eqref{eq: MILP} is infeasible.
  			\If {$\gamma^TV_C^k-\gamma^TV_C^{k-1}\leq\epsilon$}
  			\State $\mathcal{V}\leftarrow \mathcal{V}\cup\{\gamma^T V_C^k\}$
  			\EndIf
  			\EndFor
  			\If{$\mathcal{V}=\emptyset$}
  			\State Return to step $2$
  			\Else 
  			\State $t^*\leftarrow\argmax\{V_{C,t}:t=1,2,\cdots,|\mathcal{H}|\}$
  			\State $\mu\leftarrow$ policy obtained from $V_{C,t^*}$
  			\State \Return $\mu$
  			\EndIf
  		\end{algorithmic}
    \end{algorithm}
\end{center}

As shown in Algorithm \ref{alg: HVI}, we first initialize an arbitrary set of initial policies using sampling approach, where the sample space is the product of $|S_\mathcal{P}|$ probability simplices in $\mathbb{R}^{|U_C|}$. Then by solving the optimal control problem from the perspective of adversary on the MDP induced by each control policy \cite{bhatia2010sampling,lahijanian2012temporal,wolff2012robust,ding2014optimal}, we can solve for a set of expected rewards for the controller $\mathcal{H}=\{V_{C,1},\cdots,V_{C,t},\cdots,V_{C,|\mathcal{H}|}\}$ associated with the initial policies. 

For each initial expected reward $V_{C,t}\in\mathcal{H}$, value iteration (line $3$ to line $7$) is used to find a control policy such that the objective function $\gamma^T V_C$ is maximized. In particular, at iteration $k+1$, given the expected reward obtained from previous iteration $V_C^{k}$, the following mixed integer linear programming (MILP) is solved to calculate the proper control policy $\mu^{k+1}$.
\begin{align}
\underset{\mu,\lambda,V_C,V_A}{\max}\quad&\gamma^TV_C\label{eq: MILP}\\
\st\quad&V_C(s_\mathcal{P})\leq B_C^k(s_\mathcal{P},\mu,u_A)\nonumber\\
&\quad\quad\quad+(1-\lambda(s_\mathcal{P},u_A)Z,~\forall s_\mathcal{P},u_A\nonumber\\
&B_A^k(s_\mathcal{P},\mu,u_A)\leq V_A(s_\mathcal{P})\nonumber\\
&\leq B_A^k(s_\mathcal{P},\mu,u_A)+(1-\lambda(s_\mathcal{P},u_A)Z,~\forall s_\mathcal{P},u_A\nonumber\\
&\eqref{eq: anchoring bias}~\eqref{eq: constraint 1}~\eqref{eq: constraint 2}~\eqref{eq: constraint 3}~\eqref{eq: constraint 4}\nonumber
\end{align}
where $B_C^k(s_\mathcal{P},\mu,u_A)$ and $B_A^k(s_\mathcal{P},\mu,u_A)$ are obtained by \eqref{eq: BC} and \eqref{eq: BA} using $V_C^k(s_\mathcal{P})$ and $V_A^k(s_\mathcal{P})$, respectively. Note that when solving the MILP \eqref{eq: MILP}, the policy chosen by the adversary is the best response to $\tilde{\mu}^k$ obtained from \eqref{eq: anchoring bias}. The algorithm terminates when either $V_C^k-V_C^{k-1}\leq\epsilon$ or the MILP \eqref{eq: MILP} is infeasible. The first termination condition focuses on the scenario where an optimal $V_C$ can be found by solving the optimization problem. Since the initial guess is given arbitrarily while $V_C$ is bounded within $[0,\sum_{\phi\in\Phi}r(\phi)]$, thus MILP \eqref{eq: MILP} might be infeasible. In this case, such an initial guess should be skipped and the value iteration module terminates. After a feasible $V_C$ is found at some iteration $t$, we store $V_C$ in vector $\mathcal{V}$. Then the control policy returned by Algorithm \ref{alg: HVI} is the control policy corresponding to the maximum value in $\mathcal{V}$. 

The convergence of Algorithm \ref{alg: HVI} is presented in the following theorem.
\begin{theorem}\label{theorem: convergence}
 Algorithm \ref{alg: HVI} converges in finite time.
\end{theorem}

Before presenting the proof of Theorem \ref{theorem: convergence}, we first introduce two operators denoted as $T_\mu:\mathbb{R}^{|S_\mathcal{P}|}\rightarrow\mathbb{R}^{|S_\mathcal{P}|}$ and $T:\mathbb{R}^{|S_\mathcal{P}|}\rightarrow\mathbb{R}^{|S_\mathcal{P}|}$ as follows:
\begin{align}
&T_\mu V_C(s_\mathcal{P}) =\min_{\lambda\in\mathcal{BR}(\tilde{\mu})} \sum_{u_C\in U_C}\mu(s_\mathcal{P},u_C)\sum_{u_A\in U_A}\lambda(s_\mathcal{P},u_A)\nonumber\\
&~\sum_{s_\mathcal{P}^\prime}\big[Pr_\mathcal{P}(s_\mathcal{P},u_C,u_A,s_\mathcal{P}^\prime)(W(s_\mathcal{P},u_C,u_A,s_\mathcal{P}^\prime)+V_C(s_\mathcal{P}^\prime))\big],\label{eq: normal operator}\\    
&TV_C(s_\mathcal{P})=\max_\mu\min_{\lambda\in\mathcal{BR}(\tilde{\mu})}\sum_{u_C\in U_C}\mu(s_\mathcal{P},u_C)\sum_{u_A\in U_A}\lambda(s_\mathcal{P},u_A)\nonumber\\
&~\sum_{s_\mathcal{P}^\prime}\big[Pr_\mathcal{P}(s_\mathcal{P},u_C,u_A,s_\mathcal{P}^\prime)(W(s_\mathcal{P},u_C,u_A,s_\mathcal{P}^\prime)+V_C(s_\mathcal{P}^\prime))\big],\label{eq: optimal operator}    
\end{align}
The following lemmas characterizes the operator $T_\mu$.

\begin{lemma}\label{lemma: monotonicity}
For any vectors $V$ and $V^\prime$ such that $V\leq V^\prime$, we have $T_{\mu}^kV\leq T_\mu^kV^\prime$ for all policies $\mu$ and $k$, where $T_\mu^k(\cdot)$ iteratively applying $T_\mu$ operator $k$ times.
\end{lemma}
\begin{proof}
By definition \eqref{eq: normal operator} and \eqref{eq: optimal operator}, we can view the operator $T_\mu^k$ as the total expected reward collected from a $k$-stage problem with cost per stage $\tilde{W}^k(s_\mathcal{P})$. Increasing $V$ is equivalent to increasing the terminal reward (e.g., the reward collected when reaching the destination) in the $k$-stage problem. Since cost per stage is fixed, hence increasing $V$ will increase the expected total reward in the $k$-stage problem, which implies monotonicity of $T_\mu^k$.
\end{proof}
We omit the proof for Lemma \ref{lemma: monotonicity} due to space limit.

\begin{lemma}\label{lemma: convergence}
Denote the expected reward induced by proper control policy $\mu$ and adversary policy $\lambda\in\mathcal{BR}(\tilde{\mu})$ as $V_C^{\mu,\lambda}$. Then $V_C^{\mu,\lambda}$ satisfies $\lim_{M\rightarrow\infty}(T_{\mu}^MV_C)=V_C^{\mu,\lambda}$.
\end{lemma}
\begin{proof}

Since we focus on stationary policies, then by inducting Lemma \ref{lemma: value bound}, $T_{\mu}^MV_C$ can be represented as
\begin{equation}\label{eq: operator with power}
T_\mu^MV_C={Pr}^MV_C +\sum_{m=0}^{M-1}{Pr}^m\tilde{W},
\end{equation}
where $Pr$ is the transition matrix of the Markov chain induced by control policy $\mu$ and adversary policy $\lambda$. Since the control policy $\mu$ is proper, we can eventually reach the set of destination states with probability $1$. By definition \eqref{eq: weight function}, no reward can be collected when starting from destination states. Therefore, we have $\lim_{M\rightarrow\infty}Pr^MV_C=0$. Then, by taking limit on both sides of \eqref{eq: operator with power} as $M$ tends to infinity, we have $\lim_{m\rightarrow\infty}T_\mu^MV_C=\lim_{M\rightarrow\infty}\sum_{m=0}^{M-1}Pr^m\tilde{W}$. By the definition of $V_C^{\mu,\lambda}$, we have $\lim_{M\rightarrow\infty}(T_{\mu}^MV_C)=V_C^{\mu,\lambda}$, and hence Lemma \ref{lemma: convergence} is proved.
\end{proof}

Finally, we have the following proposition.
\begin{proposition}\label{prop: equality}
The optimal expected total reward for the controller at each iteration $k$ satisfies $V_C^k=TV_C^{k-1}$. 
\end{proposition}
\begin{proof}
Suppose the expected reward for the controller is $\bar{V}_C^k$ at some iteration $k$ such that $\bar{V}_C^k\neq TV_C^{k-1}$. If $\bar{V}_C^k> TV_C^{k-1}$, we have that $\bar{V}_C^k$ is not a feasible solution to MILP \eqref{eq: MILP}. If $\bar{V}_C^k< TV_C^{k-1}$, then starting from $\bar{V}_C^k$, we can always search along some direction in the feasible region of \eqref{eq: MILP} until we reach the boundary of the feasible region to find some $\hat{V}_C^k\geq \bar{V}_C^k$. Hence, $\bar{V}_C^k$ is not the optimal solution to \eqref{eq: MILP}. Therefore, we have $V_C^k=TV_C^{k-1}$ holds.  
\end{proof}
In the following, we present the proof of Theorem \ref{theorem: convergence}.
\begin{proof}
\emph{(Proof of Theorem \ref{theorem: convergence}.)} We show that Algorithm \ref{alg: HVI} terminates within finite iterations because both outer and inner loops terminate within finite iterations. 

First, the outer loop executes exactly $|\mathcal{H}|$ times and thus the outer loop terminates within finite iterations. 

Next, we show at each outer loop iteration $t$, the value iteration module converges within finite time. It is obvious that the inner loop terminates when the initial guess on $V_C$ is not feasible. In the following we focus on the feasible case. Let $k$ be the iteration index of value iteration (line $3$ to line $7$). Let us denote the expected reward of the controller induced by control policy $\mu^k$ and adversary policy $\lambda^k\in\mathcal{BR}(\tilde{\mu^k})$ at each iteration $k$ as $V_C^k$. Let the expected reward of each transition starting from state $s_\mathcal{P}$ and the transition matrix under control policy $\mu^k$ and adversary policy $\lambda^k$ be $\tilde{W}^k(s_\mathcal{P})=\sum_{u_C}\sum_{u_A}\sum_{s_\mathcal{P}^\prime}\mu^k(s_\mathcal{P},u_C)\lambda^k(s_\mathcal{P},u_A)W(s_\mathcal{P},u_C,u_A,s_\mathcal{P}^\prime)$ and $Pr^k(s_\mathcal{P},s_\mathcal{P}^\prime)=\sum_{u_C}\mu^k(s_\mathcal{P},u_C)\sum_{u_A}\lambda^k(s_\mathcal{P},u_A)\allowbreak Pr(s_\mathcal{P},u_C,u_A,s_\mathcal{P}^\prime)$, respectively. By Lemma \ref{lemma: value bound} and Proposition \ref{prop: equality}, we observe that $V_C^{k+1}=TV_C^k$ is equivalent to find a control policy $\mu^{k+1}$ such that $T_{\mu^{k+1}}V_C^{k}=TV_C^k$. Therefore $V_C^k=T_{\mu^k}V_C^k=\tilde{W}^k+Pr^kV_C^k\leq \tilde{W}^{k+1}+Pr^{k+1}V_C^k= T_{\mu^{k+1}}V_C^k$, where the inequality holds by definition \eqref{eq: normal operator} and \eqref{eq: optimal operator}, i.e., $T_{\mu^k}V_C^k\leq TV_C^k$. View $V_C^k$ as $T_{\mu^{k+1}}^0V_C^k$. Then composing $T_{\mu^{k+1}}$ $m$ times and taking the limit as $m\rightarrow\infty$, by Lemma \ref{lemma: monotonicity}, we can construct a sequence of inequalities $V_C^k\leq T_{\mu^{k+1}}V_C^k,T_{\mu^{k+1}}V_C^k\leq T_{\mu^{k+1}}^2V_C^k,\cdots,T_{\mu^{k+1}}^{m-1}V_C^k\leq T_{\mu^{k+1}}^{m}V_C^k$. Therefore, we have $V_C^k\leq\lim_{m\rightarrow\infty}T_{\mu^{k+1}}^mV_C^k=V_C^{k+1}$, where the convergence of $T_\mu^m$ follows from Lemma \ref{lemma: convergence}. Hence, the expected reward increases with respect to the number of iterations $k$. Since $V_C$ is upper bounded by $\sum_{\phi\in\Phi}r(\phi)$, we claim that the value iteration module converges within finite time.
\end{proof}

\begin{figure*}[t!]
\centering
                 \begin{subfigure}{.65\columnwidth}
                 \includegraphics[width=\columnwidth]{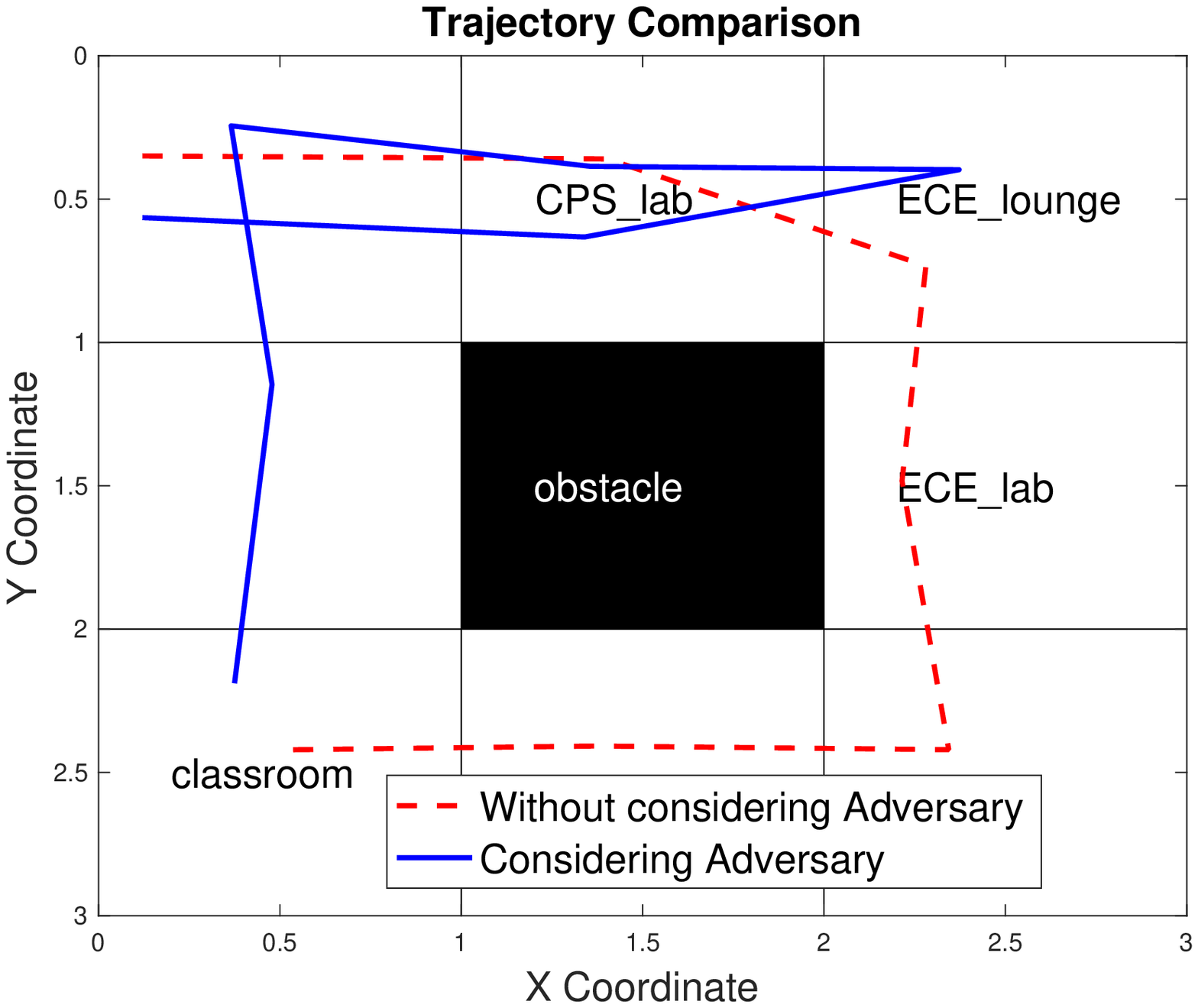}
                 \subcaption {}
                 \label{fig:trajectory}
                 \end{subfigure}\hfill
                 \begin{subfigure}{.65\columnwidth}
                 \includegraphics[width=\columnwidth]{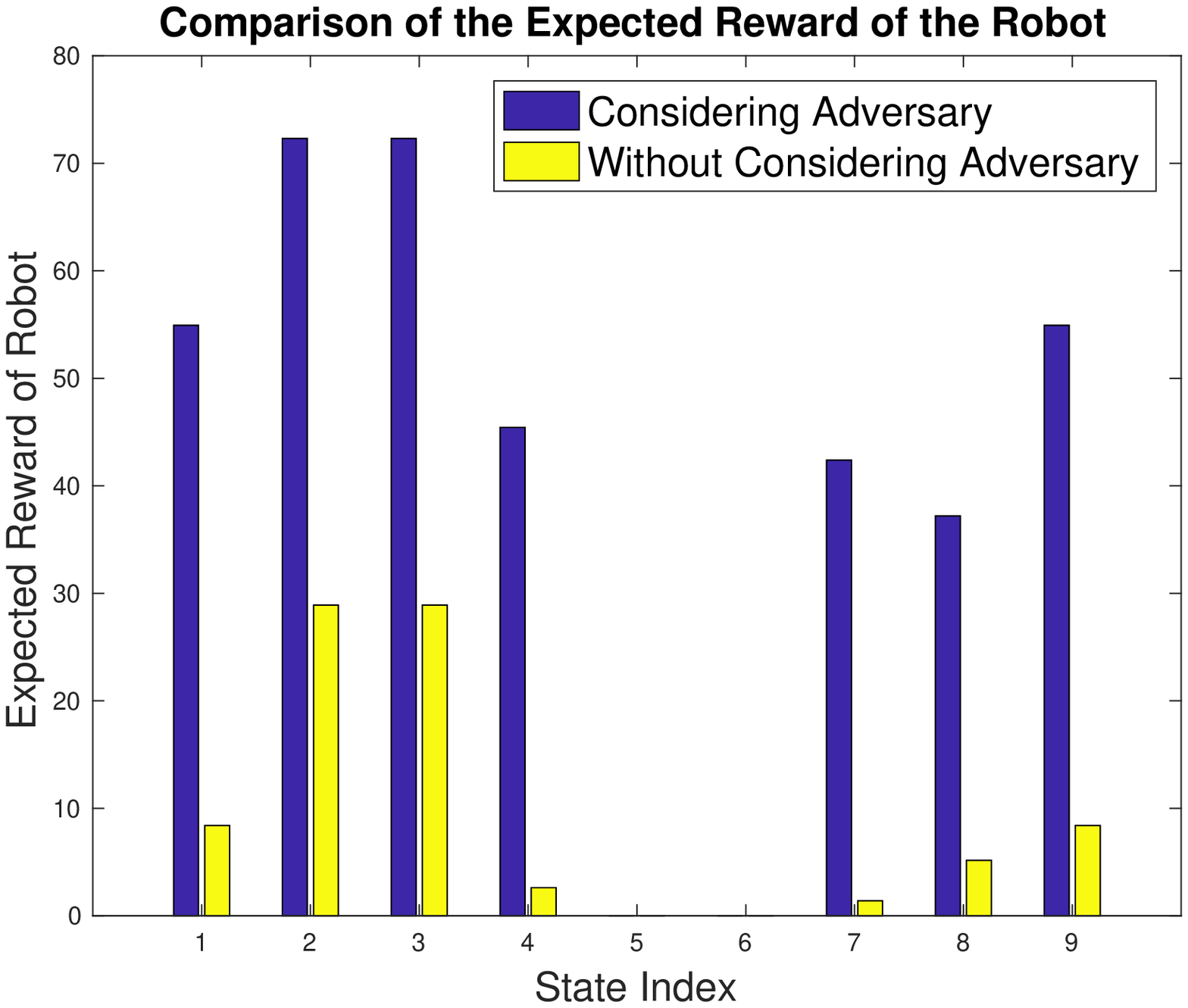}
                 \subcaption {}
                 \label{fig:improve}
                 \end{subfigure}\hfill
                 \begin{subfigure}{.65\columnwidth}
                 \includegraphics[width=\columnwidth]{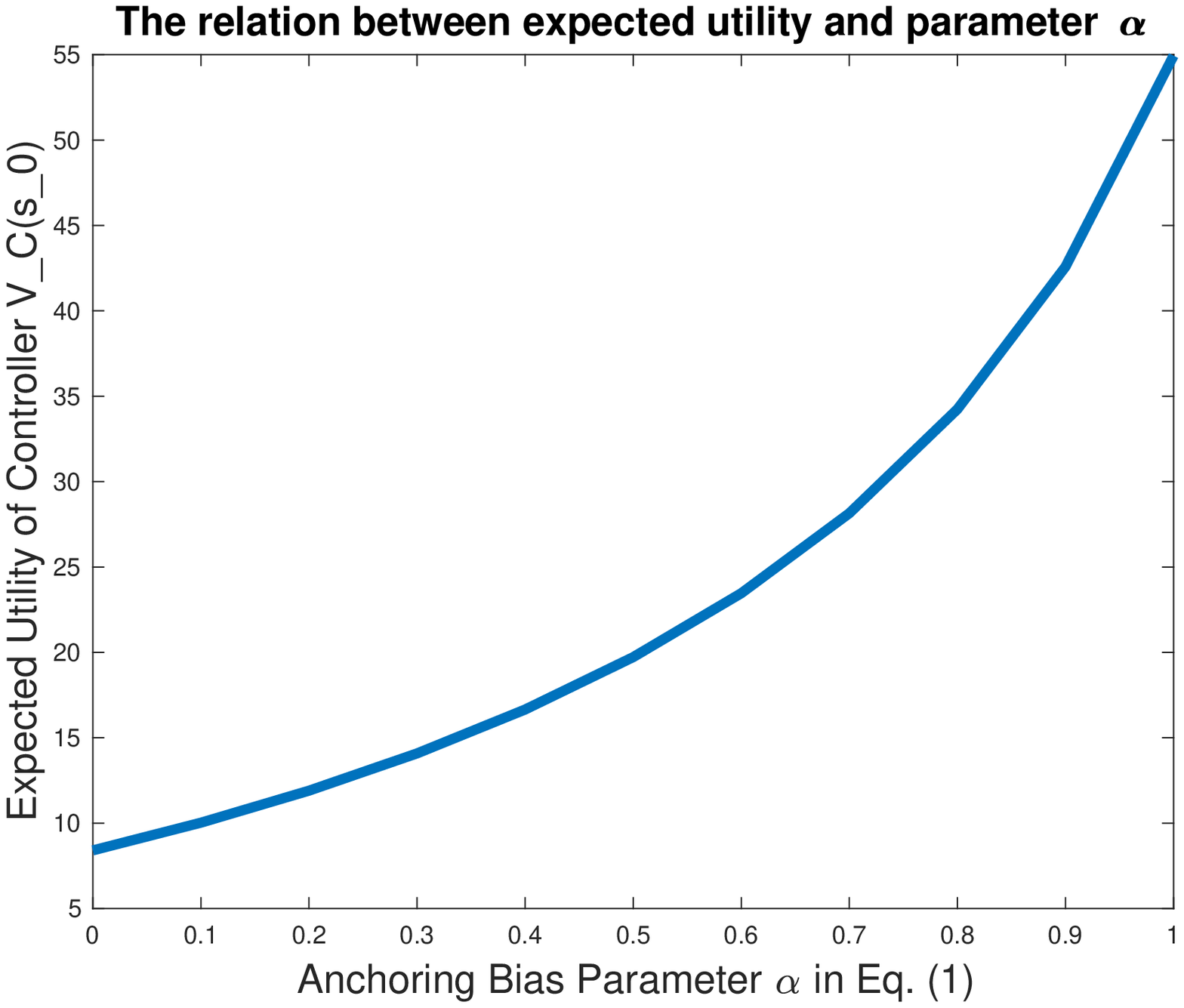}
                 \subcaption{}
                 \label{fig:observation}
                 \end{subfigure}\hfill
\caption{Fig. \ref{fig:trajectory} shows the comparison of the trajectories. The trajectory in solid line is generated using the proposed approach. The trajectory in dotted line is generated without considering the presence of adversary. Fig. \ref{fig:improve} shows the comparison of expected rewards obtained using different approaches. The blue bars are generated using the proposed approach. The yellow bars are generated without considering the presence of adversary. Fig. \ref{fig:observation} shows the relationship between controller's expected reward and anchoring bias parameter.}
\end{figure*}

Furthermore, we characterize the value function returned by Algorithm \ref{alg: HVI} using the following proposition.
\begin{proposition}\label{prop: optimal value}
The expected reward of the controller returned by Algorithm \ref{alg: HVI} is the value function obtained by committing to control strategy $\mu$ returned by Algorithm \ref{alg: HVI}.
\end{proposition}

The advantage of Algorithm \ref{alg: HVI} is that it significantly reduces computation and memory cost comparing to global optimization techniques \cite{burer2012non} and discretization-based approximate algorithms \cite{vorobeychik2012computing}. Global optimization techniques, for example, spatial branch and bound has been demonstrated non-efficient comparing to MILP. The approximate solution proposed in \cite{vorobeychik2012computing} introduces extra binary variables and constraints, whose sizes are linear to the discretization resolution. The introduction of extra variables and constraints weakens its scalability, especially for the large state space in product SG. In contrast, Algorithm \ref{alg: HVI} introduces no additional variables when solving the MILP. Therefore, Algorithm \ref{alg: HVI} significantly saves memory and model construction time for commercial solvers. Algorithm \ref{alg: HVI} does not guarantee that a global optimal solution will be found. Hence, executing Algorithm \ref{alg: HVI} from different initial points can improve the performance of Algorithm \ref{alg: HVI}.


%% file: simulation.tex
\section{Case Study}\label{sec: simulation}

In this section, we present a numerical case study to demonstrate the proposed approach.

\subsection{Case Study Settings}

Suppose a robot is performing tasks modeled in scLTL in a bounded environment. We consider the robot following standard discrete time model $x(t+1) = x(t) + \left(u_{C}(t) + u_{A}(t)+\vartheta(t)\right)\Delta t$, where $x(t)\subset\mathbb{R}^2$ is the location of the robot at time $t$, $u_C(t)\in\mathcal{U}\subset\mathbb{R}^2$ is the control input from the controller, $u_A(t)\in\mathcal{A}\subset\mathbb{R}^3$ is the input signal from the adversary and $\vartheta(t)\subset\mathbb{R}^2$ is the stochastic disturbance, $\Delta t=t_{k+1}-t_k$ is the time interval. Therefore, we have that the control signal of the robot is compromised by the adversary. Here we let $\mathcal{A}\subset\mathcal{U}$. 

We divide the region into $9$ sub-regions with each size is $1m\times1m$. We abstract the stochastic game as follows \cite{niu2018Secure}. Let each sub-region be a state in the stochastic game. Hence, the stochastic game has $9$ states and we will refer to state and sub-region interchangeably in the following. Each state can be mapped to a subset of atomic propositions by labeling function $\mathcal{L}$ as shown in Fig. \ref{fig:trajectory}. The action sets for the controller and adversary are defined as $U_C=U_A=\{N,S,W,E\}$, implying moving towards the adjacent sub-region. When the adversary compromises the control input, the probability that the robot transits to its intended state is $0.6$. Moreover, when the robot is at $ECE\_lab$, the adversary can block all the transitions of the robot (e.g., close the door of the room).

\begin{table}[b!]
\centering
\begin{tabular}{|c|c|}
  \hline
   formula & $r(\phi_i)$ \\
   \hline
  $G\neg obstacle\land F(CPS\_lab\land(F ECE\_lab\land F classroom))$ & 50 \\
  \hline
  $G\neg obstacle\land F(CPS\_lab\land F classroom)$ & 20 \\
  \hline
  $G\neg obstacle\land F(ECE\_lab\land F classroom)$ & 20 \\
  \hline
  $CPS\_lab\implies F ECE\_lounge$ & 10\\
  \hline
\end{tabular}
\caption{Specifications given to the robot. The specifications are indexed from $\phi_1$ to $\phi_4$ from top to bottom.}
\label{table: specification}
\end{table}


Suppose the robot is given $4$ specifications $\Phi=\{\phi_1,\phi_2,\phi_3,\phi_4\}$ as shown in Table \ref{table: specification}. The robot is required to visit the $CPS\_lab$ or $ECE\_lab$ before visiting $classroom$. Moreover, they are required to be visited in this particular order if possible. In the meantime, the robot should avoid $obstacle$ during the visit to guarantee safety property. Finally, the robot is required to eventually visit $ECE\_lounge$ once it has visited $CPS\_lab$.


\subsection{Case Study Results}
Let the upper left state in Fig. \ref{fig:trajectory} be the initial state. Fig. \ref{fig:trajectory} shows two trajectories generated using the proposed approach and the control policy synthesized without considering the presence of the adversary. Without considering the adversary, the control policy attempts to satisfy all the specifications in $\Phi$. However, the adversary is capable to block all the transitions at state marked as $ECE\_lab$. Therefore, following this policy can only satisfy specification $\phi_4$. Our proposed approach 
takes the potential impacts from the adversary into consideration. By using the proposed approach, specifications $\phi_2$ and $\phi_4$ are satisfied. Hence, the robot can obtain higher reward by using the proposed approach. The increment of the robot's expected reward achieved using our proposed approach is shown in Fig. \ref{fig:improve}.

In Fig. \ref{fig:observation}, we investigate the relationship between observation capability of the adversary and expected reward of the controller. We vary $\alpha$ in \eqref{eq: anchoring bias} from $0$ to $1$. When $\alpha=0$, the adversary has unlimited observation capability and it has perfect knowledge of the controller's strategy. When $\alpha=1$, the adversary makes no observation over the controller's strategy and it assumes the adversary plays uniform strategy. From Fig. \ref{fig:observation}, we observe that the expected reward of the controller increases with respect to the reduction of adversary's observation capability. Hence the more observations the adversary makes, the lower expected reward the controller obtains.

%% file: conclusion.tex
\section{conclusion}\label{sec: conclusion}

In this paper, we have investigated minimum violation problem on stochastic system in the presence of an adversary. The system is given a set of specifications modeled in scLTL. We model the interaction between the controller and adversary using a stochastic Stackelberg game. Moreover, to model the behavior of human adversaries, we consider anchoring bias. We rely on the concept of Stackelberg equilibrium to synthesize a control strategy. An efficient heristic algorithm is proposed to compute the control policy. We show the proposed algorithm converges in finite time and demonstrate the proposed approach using a numerical case study.